\documentclass[10pt]{article}
\usepackage[preprint]{tmlr}

\usepackage{amsmath,amssymb,amsthm}
\usepackage{mathtools}
\usepackage{bm}
\usepackage{booktabs}
\usepackage{graphicx}
\usepackage{hyperref}
\usepackage{url}
\usepackage{cleveref}
\usepackage{enumitem}
\usepackage{tikz}
\usetikzlibrary{arrows.meta,positioning}

\def\month{MM}
\def\year{YYYY}
\def\openreview{\url{https://openreview.net/forum?id=XXXX}}

\newtheorem{proposition}{Proposition}

\newcommand{\R}{\mathbb{R}}
\newcommand{\C}{\mathbb{C}}
\newcommand{\diag}{\operatorname{diag}}

\begin{document}

\title{Cumsum-Composable Phase Transport\\for Low-Cost Streaming Keyword Spotting}

\author{\name Mahesh Godavarti \\
  \addr A Carrot, Inc}

\maketitle

\begin{abstract}
State-space sequence models are attractive for streaming speech because they carry a compact state through time, but their training cost can depend on scan-style kernels whose constant factors are unfavorable for short audio tasks and commodity GPUs.
We study a simpler streaming-native alternative for keyword spotting: phase-transport layers whose products reduce to cumulative sums.
Each layer maps frame features to complex channels, transports them by learned rotations, accumulates over a finite window by prefix differences, and then applies a gated residual update.
The same representation has an exact batched training form and an exact online inference form: batch training uses ordinary cumulative sums, while streaming inference maintains the same prefix state with one update per frame and recovers rolling windows by prefix differences.
The key structural assumption is unitary transport: inverse rotations used to build the prefix state are norm-preserving, so finite memory is introduced by window or block readouts rather than by ill-conditioned inverse decay factors.

On Google Speech Commands v2 with 12 labels, mel+cumsum models are competitive with our compact convolutional baseline.
The strongest single-seed mel+cumsum run reaches 97.3\% test accuracy; a 51.6K-parameter tied variant reaches the same value, and a 24.8K-parameter tied model reaches 96.8\%, comparable to a 25.6K-parameter MelCNNMaxPool baseline at 97.1\%.
In a matched cumsum-versus-scan benchmark with the same mel front end, tied projections, normalization, and gating, cumsum+window gives comparable test accuracy, 94.82\% versus 94.33\% for a learned-decay scan, while training 1.07x faster and reducing single-example latency from 7.09 ms to 5.01 ms on a Tesla T4.
An optimized raw-audio cumsum front end matches mel front-end single-example latency in this setup, 337 microseconds versus 359 microseconds.
Overall, cumsum-composable local transport retains competitive keyword-spotting accuracy while replacing scan-style temporal aggregation with ordinary cumulative sums, prefix differences, and a small streaming state.
\end{abstract}

\section{Introduction}

Streaming speech systems need to preserve temporal evidence while processing audio under tight latency and cost constraints.
Modern state-space models (SSMs) are one answer: they replace all-pairs attention over long histories with a compact recurrent state, and selective variants can choose what to retain or forget as input arrives~\citep{gu2021s4,smith2023s5,gu2023mamba}.
For short-window keyword spotting, however, the practical tradeoff is different.
A one-second command is a finite local pattern, and compact CNNs already map well to commodity hardware~\citep{zhang2017helloedge,choi2019tcresnet,kim2021bcresnet}.
If an SSM-like model requires scan kernels, repeated state updates, and nontrivial backward passes, its asymptotic linearity may not translate into low training cost on small tasks.

This paper studies a cumsum-composable alternative.
The central operation is a transported finite-window sum.
Given complex value channels $z_t \in \C^n$ and phase operators $U_t = \diag(e^{i\phi_{t,1}},\ldots,e^{i\phi_{t,n}})$, define
\begin{equation}
  y_t
  =
  U_t
  \sum_{\tau=\max(1,t-W+1)}^t U_\tau^{-1} z_\tau .
  \label{eq:windowed-transport}
\end{equation}
This is a local phase-transport analogue of a recurrent memory: evidence from time $\tau$ is carried into the coordinate frame at time $t$ before aggregation.
When the phases are fixed or are cumulative sums of per-frame increments, \cref{eq:windowed-transport} is computed by prefix sums and prefix differences, one of the standard primitives for parallel sequence computation~\citep{blelloch1990prefix}.
It is therefore easy to train in parallel and easy to run in streaming mode.
The transport makes this more expressive than an ordinary cumulative sum over frames.
The contribution of frame $\tau$ to output time $t$ is $U_tU_\tau^{-1}z_\tau$, so the layer computes learned relative-phase evidence as a function of lag while still admitting a prefix representation.

The experiments in this paper are deliberately modest.
We use Google Speech Commands v2~\citep{warden2018speechcommands}, a short-command benchmark, to ask whether this architecture is worth studying before scaling to full ASR.
Fixed-frequency mel+cumsum models are strong and efficient.
The contribution is a practical architectural result: finite-memory unitary transport gives a useful class of streaming sequence layers whose recurrence factors exactly into prefix accumulators.
The experiments show that this simpler batch and streaming implementation remains accurate enough to be viable on this task.

\paragraph{Contributions.}
We make six claims, all scoped to the local experiments reported here.
\begin{itemize}[leftmargin=*]
\item We formulate a phase-transport layer whose finite-window representation is computed by cumulative sums and has an exact streaming equivalent.
\item We explain why unitary transport is the structural assumption that makes the prefix factorization numerically well conditioned: the inverse transports remain bounded.
\item We describe the corresponding streaming front end: audio is converted to low-latency mel or learned spectral frames, each frame updates the cumsum state once, and rolling keyword scores are emitted from finite-window states.
\item We show that stacked local mel+cumsum layers are competitive with the compact CNN baseline in our local Speech Commands setup, reaching up to 97.3\% test accuracy.
\item We show that tied local cumsum layers retain strong accuracy with small parameter counts, including a 24.8K-parameter model at 96.8\%.
\item We report a matched systems benchmark against a learned-decay scan layer, showing lower latency and comparable accuracy without a custom scan kernel.
\end{itemize}

\section{Cumsum Phase Transport}

\subsection{Windowed Transport by Prefix Differences}

Let $z_t \in \C^n$ be the complex value produced from an input feature vector at time $t$.
Let $U_t \in \C^{n \times n}$ be a diagonal unitary transport operator.
The transported finite-window sum in \cref{eq:windowed-transport} can be written with a prefix accumulator
\begin{equation}
  p_t = \sum_{\tau=1}^t U_\tau^{-1} z_\tau .
\end{equation}
Then
\begin{equation}
  y_t = U_t \left(p_t - p_{t-W}\right),
  \qquad p_s = 0 \text{ for } s \leq 0 .
  \label{eq:prefix-window}
\end{equation}
This is the entire computational advantage: the layer uses a cumulative sum of transported values, a delayed prefix subtraction, and an elementwise rotation back to the current frame.

\begin{proposition}[Exact batch/streaming equivalence]
For fixed model parameters and a fixed input sequence, \cref{eq:prefix-window} produces the same $y_t$ whether $p_t$ is computed by a batched cumulative sum over the whole sequence or by the streaming update
\[
  p_t = p_{t-1} + U_t^{-1}z_t .
\]
\end{proposition}

\begin{proof}
The streaming recurrence expands to $p_t=\sum_{\tau=1}^t U_\tau^{-1}z_\tau$, which is exactly the batched cumulative sum.
Substituting this prefix into \cref{eq:prefix-window} gives the same finite-window transported sum.
\end{proof}

\subsection{Why Transport Before Accumulation?}

An ordinary cumulative sum aggregates values in a fixed coordinate system.
Temporal order enters only through window membership and through nonlinear layers after the sum.
Phase transport changes the operation: the contribution from time $\tau$ to time $t$ is
\[
  U_tU_\tau^{-1}z_\tau .
\]
For fixed learned phases $\phi_{t,k}=t\omega_k$, the $k$th complex coordinate is weighted by $e^{i(t-\tau)\omega_k}$.
The layer therefore implements learned oscillatory kernels over relative lag, but all such kernels share a cumsum implementation because each term can be moved into a common prefix coordinate system before accumulation.
This is the modeling role of complex rotations: they make the accumulated evidence depend on relative phase, not only on the unordered sum of local values.
Complex exponentials and rotations are standard tools in sequence models, appearing in sinusoidal and rotary position representations as well as diagonal complex state-space parameterizations~\citep{vaswani2017attention,su2021roformer,gupta2022dss,gu2022s4d}.
The contribution here is narrower: a finite-window, unitary transport parameterization whose online recurrence and batched training computation are exactly the same prefix representation.

\subsection{Streaming Front End}

The cumsum layer is naturally paired with a streaming acoustic front end.
At deployment, incoming waveform samples are held in a short analysis buffer.
Every hop, for example 80 samples (5 ms) or 160 samples (10 ms) at 16 kHz, the front end emits one acoustic frame.
That frame can be a conventional log-mel vector, as in the strongest experiments here, or a learned cumsum spectral vector as in the learned-front-end experiments.

The online pipeline is:
\[
  x_{1:t}
  \;\to\;
  m_t
  \;\to\;
  z_t
  \;\to\;
  p_t = p_{t-1} + U_t^{-1}z_t
  \;\to\;
  y_t = U_t(p_t-p_{t-W})
  \;\to\;
  s_t .
\]
Here $m_t$ is the acoustic frame, $z_t$ is the complex value injected into the transport layer, $p_t$ is the maintained prefix state, $y_t$ is the current finite-window transported state, and $s_t$ is a keyword logit or trigger score.
The model therefore does not need to wait for a full one-second clip.
It can update a rolling score at the frame hop rate and emit when the score crosses a task-specific threshold.
Continuous-stream keyword spotting normally requires additional trigger metrics such as false accepts, false rejects, and detection latency~\citep{rybakov2020streaming,jose2022latency}; the present experiments evaluate the model first as a clip classifier with an exact streaming implementation.
\Cref{fig:streaming-layer} summarizes the resulting streaming layer.

The memory needed per layer is small.
For a hard window, the implementation stores the current prefix $p_t$ and a ring buffer of old prefixes needed to form $p_t-p_{t-W}$.
For block decay, it stores block summaries and updates them at block boundaries.
In both cases, online inference uses the same representation as batched training; the difference is only whether prefixes are produced by a full-sequence cumulative sum or by the one-frame recurrence.

\subsection{Fixed Learned Phases}

The fixed-frequency models use one learned frequency per complex channel and layer:
\begin{equation}
  \phi_{t,k} = t\omega_k .
\end{equation}
This gives each complex channel a learned local oscillator.
The resulting rotations are unitary, so the transport part is norm-preserving; finite memory is supplied by the window or block-decay readout rather than by continuously damping the oscillator.
Because the phases are deterministic functions of time within a layer, the whole sequence can be processed by cumulative sums during training and by a one-step prefix update during streaming inference.

\subsection{Why Unitary Transport?}

The prefix construction requires applying an inverse transport to each incoming value before accumulation.
For unitary rotations, $U_t^{-1}=U_t^*$ and $\|U_t^{-1}z\|=\|z\|$, so the prefix terms remain on the same numerical scale as the projected values.
This is the algebraic reason that batch cumsums and streaming prefix updates are both simple and stable.

The contrast with decaying transport clarifies the design choice.
Suppose one tried to use $A_t=\rho^tU_t$ with $0<\rho<1$ in the same prefix form.
The prefix accumulator would contain terms $A_\tau^{-1}z_\tau=\rho^{-\tau}U_\tau^{-1}z_\tau$.
Although the final readout could multiply by a small forward factor and cancel this scaling in exact arithmetic, the prefix state itself would involve division by progressively smaller numbers.
In finite precision this can create large scale separation, overflow risk, and loss of low-order information.
Unitary transport avoids this class of conditioning issues because the inverse used to build prefixes has norm one.
For this reason, forgetting is introduced outside the transport inversion, through hard-window subtraction or block-level weighting.

\subsection{Layer Form}

The main mel+cumsum layer follows the implementation in the local experiment code.
A log-mel frame $m_t \in \R^{40}$ is embedded into $h_t \in \R^d$.
At each cumsum layer, a linear projection creates real and imaginary channels:
\[
  (a_t,b_t) = P_\ell h_t,
  \qquad
  z_t = a_t + i b_t .
\]
The transported window sum $y_t$ is computed by \cref{eq:prefix-window}, converted back to real coordinates, normalized, and passed through a gated residual update:
\begin{equation}
  h_t \leftarrow h_t + \operatorname{GLU}(\operatorname{BN}([\Re y_t,\Im y_t])) .
\end{equation}
For windowed classification, the model collapses the final complex channels to magnitudes and applies temporal max pooling.
Tied variants share the projection and GLU across layers while retaining per-layer frequency and normalization parameters.
\Cref{tab:fixed-architecture} summarizes the simple fixed-frequency architecture used for the main mel+cumsum results.

\begin{table}[t]
\centering
\caption{Simple fixed mel+cumsum architecture. The main experiments vary $d$, $L$, $W$, hop length, and whether the projection/GLU are tied across layers.}
\label{tab:fixed-architecture}
\small
\begin{tabular}{p{0.22\linewidth}p{0.68\linewidth}}
\toprule
Component & Implementation \\
\midrule
Acoustic front end & 40-bin log-mel spectrogram, 16 kHz audio, 400-sample FFT window, hop 80 or 160 samples; SpecAugment during training. \\
Embedding & Linear map $m_t\in\R^{40}$ to $h_t\in\R^d$. \\
Value projection & Linear map $h_t\mapsto(a_t,b_t)\in\R^{d/2}\times\R^{d/2}$, interpreted as $z_t=a_t+ib_t$. \\
Fixed phase bank & Per-layer learned frequencies $\omega_{\ell,k}$ with $\phi_{t,k}=t\omega_{\ell,k}$ and $U_t=\diag(e^{i\phi_{t,k}})$. \\
Cumsum transport & Compute $p_t=\sum_{\tau\leq t}U_\tau^{-1}z_\tau$ and read out $y_t=U_t(p_t-p_{t-W})$. \\
Layer update & Concatenate $[\Re y_t,\Im y_t]$, apply BatchNorm, GLU, and residual addition. \\
Classifier & Convert final real/imaginary channels to magnitudes, max-pool over time, and apply a linear 12-way classifier. \\
Tied variant & Share the value projection and GLU across all layers; keep frequencies and BatchNorm separate per layer. \\
\bottomrule
\end{tabular}
\end{table}

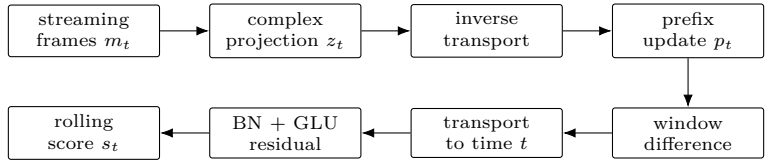
\begin{figure}[t]
\centering
\begin{tikzpicture}[
  >=Latex,
  node distance=0.65cm and 0.65cm,
  block/.style={draw, rounded corners=1pt, align=center, minimum width=2.0cm, minimum height=0.62cm, font=\scriptsize, inner sep=3pt}
]
\node[block] (frames) {streaming\\frames $m_t$};
\node[block, right=of frames] (proj) {complex\\projection $z_t$};
\node[block, right=of proj] (inverse) {inverse\\transport};
\node[block, right=of inverse] (prefix) {prefix\\update $p_t$};
\node[block, below=of prefix] (window) {window\\difference};
\node[block, left=of window] (forward) {transport\\to time $t$};
\node[block, left=of forward] (glu) {BN + GLU\\residual};
\node[block, left=of glu] (score) {rolling\\score $s_t$};
\draw[->] (frames) -- (proj);
\draw[->] (proj) -- (inverse);
\draw[->] (inverse) -- (prefix);
\draw[->] (prefix) -- (window);
\draw[->] (window) -- (forward);
\draw[->] (forward) -- (glu);
\draw[->] (glu) -- (score);
\end{tikzpicture}
\caption{Streaming phase-transport layer. Batch training computes all prefixes by a cumulative sum; streaming inference maintains the same prefix state one frame at a time.}
\label{fig:streaming-layer}
\end{figure}

\subsection{Block Decay}

Hard finite windows are efficient but can be unstable because the representation changes discontinuously when an old term exits the window.
The block-decay variant replaces a single hard window with a geometrically weighted sum of blocks:
\begin{equation}
  y_t
  =
  U_t \sum_{b=0}^{K-1} \lambda^b
  \left(p_{t-bW} - p_{t-(b+1)W}\right).
  \label{eq:block-decay}
\end{equation}
This keeps the computation cumsum-based while adding smooth forgetting.
The main mel experiments in this paper use hard windows; block decay is a cumsum-compatible alternative for settings where smoother forgetting is useful.

\section{Experimental Setup}

\paragraph{Dataset.}
All experiments use Google Speech Commands v2 with one-second 16 kHz clips.
The label set has 12 classes: ten commands (\texttt{yes}, \texttt{no}, \texttt{up}, \texttt{down}, \texttt{left}, \texttt{right}, \texttt{on}, \texttt{off}, \texttt{stop}, \texttt{go}) plus \texttt{unknown} and \texttt{silence}.
Unknown examples are subsampled to match the average number of target-command examples; silence is generated from the dataset background-noise clips.

\paragraph{Baselines.}
We compare three families.
First, compact convolutional baselines use either log-mel features or raw waveform convolutions.
The strongest mel baseline is MelCNNMaxPool, a TC-ResNet-style one-dimensional temporal CNN over log-mel channels with SpecAugment~\citep{choi2019tcresnet,park2019specaugment}.
Published KWS systems include DS-CNN-style embedded models and stronger broadcast-residual CNNs~\citep{zhang2017helloedge,kim2021bcresnet}; the local MelCNNMaxPool model serves as the controlled compact reference for the experiments below.
Second, learned spectral front ends replace the fixed FFT/mel transform with learned sinusoidal cumsum filters, connecting this part of the study to learned audio front ends such as LEAF~\citep{zeghidour2021leaf}.
Third, a matched scan baseline keeps the same mel front end, tied projections, normalization, and gating as the cumsum model, but replaces each hard-window cumsum with a learned exponential-decay scan.

\paragraph{Training.}
Full runs use the full local training set, approximately 36K examples, for 40 or 80 epochs depending on the experiment.
Smoke tests use a two-epoch subset of roughly 4K training examples and are used only for architectural deconstruction.
Mel/CNN models use Adam with learning rate $10^{-3}$ and weight decay $10^{-4}$.
All reported results in this draft are single-seed unless otherwise stated.
The cumsum-versus-scan systems benchmark was run on a Tesla T4 with CUDA 12.4, PyTorch 2.6, Torchaudio 2.6, and batch size 128 for training.

\section{Results}

\subsection{Mel+Cumsum Is Competitive}

\Cref{tab:main-results} shows the strongest local mel+cumsum and CNN results.
The strongest single-seed cumsum runs reach 97.3\% test accuracy.
A cumsum-composable temporal primitive remains comparable to the local CNN baseline while using an exact streaming prefix representation.
Tying and small local windows preserve most of the performance with much lower parameter counts.

\begin{table}[t]
\centering
\caption{Main Speech Commands results from the local architecture comparison. All numbers are single-seed test accuracies.}
\label{tab:main-results}
\begin{tabular}{llrr}
\toprule
Model & Configuration & Test acc. & Params \\
\midrule
MelCumsumFixed & W=5, 8L, n=80, hop=160, untied, 40ep & 97.3 & 160,892 \\
MelCumsumFixedTied & W=10, 8L, n=120, hop=80, 80ep & 97.3 & 51,612 \\
MelCNNMaxPool & hop=160, 40ep & 97.1 & 25,628 \\
MelCumsumFixedTied & W=5, 8L, n=100, hop=160, 40ep & 97.0 & 37,012 \\
MelCumsumBidirTied & W=10, 8L, n=80, hop=80, 40ep & 96.8 & 31,612 \\
MelCumsumFixedTied & W=10, 8L, n=80, hop=80, 80ep & 96.8 & 24,812 \\
\bottomrule
\end{tabular}
\end{table}

The 24.8K-parameter fixed tied model is the cleanest efficiency result: it is close to the 25.6K-parameter MelCNNMaxPool baseline while replacing temporal convolutions after the mel front end with cumsum-composable phase transport.
The 51.6K-parameter tied model reaches the same observed accuracy as the 160.9K-parameter untied variant.

\subsection{Small Local Windows and Depth Matter}

The final results favor repeated local windows over shallow wide aggregation.
For example, W=10 with eight tied layers reaches 96.8--97.3\% depending on width and training length, making repeated tight windows the strongest configuration in the local sweep.
\Cref{tab:window-depth-app} gives the fixed-model ablation behind this claim: at a comparable 40-frame nominal span, moving from two W=20 layers to eight W=5 layers improves test accuracy from 94.4\% to 96.2\%.
This supports the interpretation that cumsum is best used as a local temporal building block with nonlinear gating between layers, rather than as one large undifferentiated accumulator.

\subsection{Tying Preserves Performance}

The tied variants are important because they separate architectural structure from raw parameter count.
The untied W=5, eight-layer model reaches 97.3\% with 160.9K parameters.
A tied W=10, eight-layer model with n=120 reaches the same 97.3\% test accuracy with 51.6K parameters, a reduction of about 3.1x.
At nearly the same parameter count as the MelCNNMaxPool baseline, the tied W=10, n=80 model reaches 96.8\% with 24.8K parameters.
This suggests that repeated local transport with shared projections is a compact temporal mixer rather than merely a larger alternative to convolution.

\subsection{Cumsum Acoustic Front Ends}

Replacing the FFT/mel front end with sinusoidal cumsum filters gives a direct streaming acoustic front end.
\Cref{tab:frontend} shows representative results with the same CNN backend used for the mel reference.

\begin{table}[t]
\centering
\caption{Front-end comparison. Cumsum spectral front ends provide a streaming acoustic representation before the sequence model.}
\label{tab:frontend}
\small
\begin{tabular}{llrr}
\toprule
Model & Front end & Test acc. & Params \\
\midrule
MelCNNMaxPool & FFT + log-mel & 97.1 & 25,628 \\
CumsumSpecCNN & cumsum, frozen mel phases, log-mag+re+im & 95.2 & 46,828 \\
CumsumSpecCNN & cumsum, learned phases, log-mag+re+im & 93.9 & 30,508 \\
LearnedSpecCNN & cumsum, learned phases, log-mag & 94.1 & 25,668 \\
RawCNN & raw waveform Conv1d & 91.0 & 25,228 \\
\bottomrule
\end{tabular}
\end{table}

The cumsum front end is important because it shows that the same prefix-sum primitive can act before the mel stage as well as after it.
The strongest cumsum front end uses mel-spaced frozen phases and keeps phase-carrying features, reaching 95.2\% before the sequence model.
The detailed ablation in \cref{tab:cumsum-frontend-hop160-app,tab:cumsum-frontend-hop400-app} shows that retaining real and imaginary channels is especially helpful when the hop is large.
On short commands, the fixed FFT/mel/log pipeline still supplies a strong inductive bias.
The strongest architecture in this draft therefore uses cumsum transport after mel features, rather than replacing mel features entirely.

\subsection{Cumsum Versus Learned-Decay Scan}

The most direct systems comparison keeps the mel front end and model shape fixed.
Both models use 40-bin log-mel frames at hop 80, $d=80$, eight tied layers, the same projections, BatchNorm, GLU residual updates, and max-pooling readout.
The only substantive difference is the temporal primitive: MelCumsumFixed uses a hard window $W=10$ implemented by \texttt{torch.cumsum} and prefix subtraction, while MelScanFixed replaces the window with learned per-channel exponential decay implemented by a custom Triton scan.
The scan model has 320 extra parameters, one decay scalar per frequency per layer.

\begin{table}[t]
\centering
\caption{Matched cumsum-versus-scan benchmark. Both models use the same mel front end and tied eight-layer architecture.}
\label{tab:scan-benchmark}
\small
\begin{tabular}{llrrrr}
\toprule
Model & Temporal primitive & Params & Best val & Test & B=1 latency \\
\midrule
MelCumsumFixed & cumsum + hard window $W=10$ & 24,812 & 94.8 & 94.82 & 5.01 ms \\
MelScanFixed & Triton scan + learned decay & 25,132 & 94.2 & 94.33 & 7.09 ms \\
\bottomrule
\end{tabular}
\end{table}

\Cref{tab:scan-benchmark} shows comparable accuracy in this single matched run.
The primary difference is implementation cost: the cumsum model uses ordinary cumulative sums and prefix subtraction, trains in 16.2 seconds per epoch versus 17.4 seconds per epoch, and has 1.42x lower single-example inference latency.
The raw primitive benchmark gives the cleaner implementation signal.
At the mel-layer size $T=200,N=40$, \texttt{torch.cumsum} is 3.2--4.5x faster than the Triton scan across batch sizes 1--128.

\begin{table}[t]
\centering
\caption{Primitive timing at the mel-layer sequence size, $T=200,N=40$, on a Tesla T4.}
\label{tab:primitive-benchmark}
\small
\begin{tabular}{rrrr}
\toprule
Batch & \texttt{torch.cumsum} & Triton scan & Ratio \\
\midrule
1 & 84.9 us & 270.6 us & 3.2x \\
8 & 84.9 us & 273.6 us & 3.2x \\
32 & 73.4 us & 331.8 us & 4.5x \\
128 & 140.5 us & 514.5 us & 3.7x \\
\bottomrule
\end{tabular}
\end{table}

The scan model learns short effective memories.
Across layers, the learned decay values imply effective windows between roughly 3 and 21 frames.
This supports the finite-memory hypothesis: for this keyword-spotting configuration, the learned recurrent memory behaves like a short local filter, which the cumsum window supplies directly.

\subsection{Streaming Front-End Timing}

The cumsum acoustic front end is also sensitive to implementation details.
The initial raw-audio cumsum front end used a strided time dimension and performed window subtraction before downsampling, which made \texttt{torch.cumsum} the bottleneck.
Using a contiguous $(B,N,T)$ layout and indexing stride positions before downstream operations changes the layer-1 single-example latency from 5.8 ms to 284 us.
After these optimizations, the cumsum front end matches the mel front end at batch size one: 337 us for the cumsum front end versus 359 us for FFT+log+embedding.
At large training batches, the cumsum front end remains memory heavier because it materializes the full-resolution complex rotation tensor; fusing rotation, cumsum, window subtraction, and striding is the natural next implementation step.

\section{Discussion}

\paragraph{What the results support.}
The reliable finding is that cumsum-composable local transport is a simple and effective temporal aggregation layer for keyword spotting.
It works best after a mel front end, with small windows, multiple layers, gating, and weight tying.
It has an exact streaming update and a simple batched implementation.
The matched scan benchmark supports the systems argument: when the learned scan converges to short effective memories, the hard-window cumsum version has similar accuracy and lower latency in the measured setup.
This is a useful combination for low-cost streaming speech research.

\paragraph{Relation to SSMs.}
SSMs remain natural for long-context streaming models and production voice systems.
Finite-window speech tasks expose a complementary regime.
When the relevant memory is local and the transport operators are rotations, cumulative sums can recover the useful prefix/window structure without the full machinery of selective scan.
The closest SSM relatives are diagonal complex models, which also use per-channel complex dynamics to generate sequence kernels~\citep{gupta2022dss,gu2022s4d,smith2023s5}.
This work restricts the transport to unit-modulus rotations and supplies memory length explicitly through prefix differences, so the inverse transport used inside the accumulator remains bounded.

\begin{table}[t]
\centering
\caption{Positioning of cumsum phase transport relative to common streaming sequence model families.}
\label{tab:positioning}
\small
\begin{tabular}{p{0.18\linewidth}p{0.24\linewidth}p{0.27\linewidth}p{0.20\linewidth}}
\toprule
Model family & Memory assumption & Training primitive & Streaming form \\
\midrule
CNN/TCN & finite local receptive field & convolution & causal buffer \\
S4/S5 & long or decaying linear state & scan or structured convolution & recurrence \\
Mamba & input-selective state & selective scan & recurrence \\
This work & finite unitary transport & cumulative sum and prefix differences & prefix state, ring buffer \\
\bottomrule
\end{tabular}
\end{table}

\paragraph{Scope.}
The experiments focus on keyword spotting.
The next step is to measure whether the same cumsum transport layer retains its cost and latency advantages under stronger baselines, replicated runs, additional hardware, and streaming trigger metrics used in mobile KWS evaluation~\citep{rybakov2020streaming,jose2022latency}.

\section{Limitations and Next Experiments}

The current draft has several important limitations.
All reported results are single-seed.
The SSM comparisons use local research implementations.
The current systems benchmark is limited to one GPU environment and one matched architecture.
Speech Commands is a keyword-spotting benchmark.
A fuller perturbation suite would add pitch shift, reverberation, codec distortion, speaker holdout, and event-triggered onset metrics.

The immediate next experiments are:
\begin{enumerate}[leftmargin=*]
\item Replicate the best fixed tied mel+cumsum and MelCNNMaxPool runs over at least three seeds.
\item Extend wall-clock, memory, throughput, and profiler logging to CPU, Apple Silicon, and embedded/mobile targets.
\item Compare against standard optimized KWS baselines such as DS-CNN, TC-ResNet, BC-ResNet, CRNN/TCN, and small Conformer variants.
\item Separate three families cleanly: mel+cumsum, learned spectral cumsum, and raw/sample-clocked cumsum.
\item Add event-triggered command onset detection so the model is tested on latency and sparse emissions as well as clip-level accuracy.
\item Expand perturbation tests to pitch shift, reverberation, background noise, codec artifacts, speaker holdout, and controlled phase/magnitude perturbations.
\end{enumerate}

\section{Conclusion}

Cumsum-composable phase transport gives a simple way to build finite-window streaming speech representations with exact batched training and exact online updates.
On the local Speech Commands experiments, the strongest evidence supports fixed-frequency mel+cumsum models: they are competitive with the local compact CNN baseline, benefit from small local windows and depth, and retain strong performance under parameter tying.
The matched benchmark against a learned-decay scan strengthens the systems claim: in the short-memory regime observed here, the cumsum layer is simpler and lower-latency while retaining comparable accuracy.
This makes the next research question concrete: can the cumsum transport layer retain its simplicity, cost advantages, and competitive accuracy under multi-seed evaluation, optimized baselines, broader hardware, and harder acoustic distribution shifts?

\clearpage
\appendix

\section{Additional Ablations}

The appendix reports additional ablations from the local experiment notes.
All entries are single-seed Speech Commands v2 results under the same general training setup described in the main text.
The first tables focus on the simple fixed mel+cumsum architecture; the final tables summarize the cumsum acoustic front-end sweep.

\subsection{Window Size and Depth}

\Cref{tab:window-depth-app} decomposes the tied fixed model by local window size and depth.
The main pattern is that a stack of small windows is stronger than a shallow model with a wider window.
This is consistent with the intended use of the cumsum layer as a local aggregation primitive followed by nonlinear gating.

\begin{table}[ht]
\centering
\caption{Fixed tied mel+cumsum window/depth ablation. The nominal span is $L\times W$ in mel frames and is used only as a rough comparison of temporal coverage.}
\label{tab:window-depth-app}
\small
\begin{tabular}{lrrrrr}
\toprule
Configuration & Nominal span & Layers & Params & Best val & Test \\
\midrule
W=20, 2L & 40 & 2 & 23,612 & 94.6 & 94.4 \\
W=10, 4L & 40 & 4 & 24,012 & 95.7 & 95.8 \\
W=5, 8L & 40 & 8 & 24,812 & 96.3 & 96.2 \\
W=20, 4L & 80 & 4 & 24,012 & 94.9 & 95.2 \\
W=20, 8L & 160 & 8 & 24,812 & 95.0 & 95.5 \\
\bottomrule
\end{tabular}
\end{table}

\subsection{Weight Tying}

\Cref{tab:tying-app} isolates the effect of sharing the value projection and GLU across cumsum layers.
The tied model keeps per-layer frequencies and normalization, reducing parameters by about 3.4x in this setting while staying close to the untied model and the compact CNN reference.

\begin{table}[ht]
\centering
\caption{Fixed mel+cumsum weight-tying ablation for W=20, 4-layer models.}
\label{tab:tying-app}
\small
\begin{tabular}{lrrr}
\toprule
Model & Params & Best val & Test \\
\midrule
MelCumsumFixed, untied & 82,252 & 96.0 & 96.3 \\
MelCumsumFixedTied & 24,012 & 94.9 & 95.2 \\
MelCNNMaxPool reference & 25,628 & 95.4 & 95.4 \\
\bottomrule
\end{tabular}
\end{table}

\subsection{Fixed-Frequency Sweep}

\Cref{tab:fixed-sweep-app} lists representative fixed-frequency mel+cumsum variants from the later local architecture sweep.
The best tied model reaches the same 97.3\% test accuracy as the larger untied model at a substantially smaller parameter count.

\begin{table}[ht]
\centering
\caption{Representative fixed-frequency mel+cumsum sweep.}
\label{tab:fixed-sweep-app}
\small
\begin{tabular}{llrr}
\toprule
Model & Configuration & Test & Params \\
\midrule
MelCumsumFixed & W=5, 8L, n=80, hop=160, 40ep, untied & 97.3 & 160,892 \\
MelCumsumFixedTied & W=10, 8L, n=120, hop=80, 80ep & 97.3 & 51,612 \\
MelCumsumFixedTied & W=5, 8L, n=100, hop=160, 40ep & 97.0 & 37,012 \\
MelCumsumBidirTied & W=10, 8L, n=80, hop=80, 40ep & 96.8 & 31,612 \\
MelCumsumFixedTied & W=10, 8L, n=80, hop=80, 80ep & 96.8 & 24,812 \\
MelCumsumFixedTied & W=10, 8L, n=80, hop=80, 40ep & 96.5 & 24,812 \\
MelCumsumFixedTied & W=5, 8L, n=80, hop=160, 40ep & 96.2 & 24,812 \\
MelCumsumFixedTied & W=40, 2L, n=80, hop=80, 80ep & 94.5 & 23,612 \\
\bottomrule
\end{tabular}
\end{table}

\subsection{Architecture Alternatives}

\Cref{tab:architecture-alternatives-app} summarizes additional architecture variants from the same local sweep.
These runs address a natural reviewer question: whether the gain is simply from stacking many small nonlinear stages or from a more generic progressive architecture.
The tied fixed cumsum model with constant width remains stronger than the progressive cumsum-ResNet variants at similar parameter counts.
The comparison places the main model against nearby compact cumsum architectures and helps separate the fixed phase-transport design from generic compact cumsum variants.

\begin{table}[ht]
\centering
\caption{Architecture alternatives from the local mel+cumsum sweep.}
\label{tab:architecture-alternatives-app}
\small
\begin{tabular}{llrr}
\toprule
Model & Configuration & Test & Params \\
\midrule
MelCNNMaxPool & hop=160, 40ep & 97.1 & 25,628 \\
MelCumsumFixedTied & W=10, 8L, n=80, hop=80, 80ep & 96.8 & 24,812 \\
MelCumsumFixedTied & W=5, 8L, n=80, hop=160, 40ep & 96.2 & 24,812 \\
MelCumsumResNet & W=3, channels [24,24,32,48], 40ep & 94.9 & 25,380 \\
MelCumsumResNet & W=3, channels [16,24,32,48], 40ep & 94.6 & 22,852 \\
MelCumsumMagDeepTied & W=5, 8L, n=90, 40ep & 93.0 & 25,032 \\
\bottomrule
\end{tabular}
\end{table}

\subsection{Cumsum Acoustic Front End}

The main paper uses log-mel features because they remain the strongest acoustic representation in the current experiments.
The cumsum front-end experiments test a separate question: whether the same prefix-sum mechanism can replace the FFT/mel front end while preserving streaming computation from raw audio.
The front end maintains running complex sinusoidal sums over waveform samples and emits either log magnitude alone or log magnitude together with the real and imaginary channels.

\Cref{tab:cumsum-frontend-hop160-app} summarizes the standard 10 ms hop setting.
The strongest cumsum front end uses frozen mel-spaced phases and a learned linear bottleneck from $[\log |c|,\mathrm{Re}\ c,\mathrm{Im}\ c]$ to the CNN input channels.
This reaches 95.2\% test accuracy.
Within the cumsum front-end family, frozen mel-spaced phases outperform learned phases in this sweep, and retaining phase-carrying real and imaginary channels improves over log magnitude alone.

\begin{table}[ht]
\centering
\caption{Cumsum front-end ablation at hop 160 samples. All rows use 80 acoustic channels and the same CNN backend.}
\label{tab:cumsum-frontend-hop160-app}
\small
\begin{tabular}{llrrr}
\toprule
Model & Front-end features & Window & Best val & Test \\
\midrule
MelCNNMaxPool & FFT + log-mel & 400 & 96.4 & 97.1 \\
CumsumSpec, frozen phases & $[\log |c|,\mathrm{Re}\ c,\mathrm{Im}\ c]$ & 400 & 95.2 & 95.2 \\
CumsumSpec, frozen phases & $[\log |c|,\mathrm{Re}\ c,\mathrm{Im}\ c]$ & 160 & 95.2 & 94.8 \\
CumsumSpec, learned phases & $[\log |c|,\mathrm{Re}\ c,\mathrm{Im}\ c]$ & 400 & 94.0 & 93.9 \\
LearnedSpecCNN & $\log |c|$ & 400 & -- & 93.4 \\
\bottomrule
\end{tabular}
\end{table}

\Cref{tab:cumsum-frontend-hop400-app} repeats the comparison in a no-overlap hop-400 setting.
This setting makes the phase channels more visible: with learned phases, adding real and imaginary channels raises test accuracy from 90.8\% to 93.6\%; with frozen phases, it raises test accuracy from 92.0\% to 94.0\%.

\begin{table}[ht]
\centering
\caption{Cumsum front-end ablation at hop 400 samples, where adjacent acoustic frames do not overlap.}
\label{tab:cumsum-frontend-hop400-app}
\small
\begin{tabular}{llrrr}
\toprule
Model & Front-end features & Phases & Best val & Test \\
\midrule
MelCNNMaxPool & FFT + log-mel & fixed mel & 95.3 & 95.4 \\
CumsumSpec & $[\log |c|,\mathrm{Re}\ c,\mathrm{Im}\ c]$ & frozen mel & 94.1 & 94.0 \\
CumsumSpec & $[\log |c|,\mathrm{Re}\ c,\mathrm{Im}\ c]$ & learned & 94.0 & 93.6 \\
CumsumSpec & $\log |c|$ & frozen mel & 92.7 & 92.0 \\
CumsumSpec & $\log |c|$ & learned & 92.2 & 90.8 \\
\bottomrule
\end{tabular}
\end{table}

\Cref{tab:filterbank-frontend-app} adds a related learned-filterbank diagnostic.
The convolutional sin-cos filterbank uses Hann-windowed filters rather than a running rectangular cumsum window.
The strongest configuration in this sweep is the frozen mel-spaced log-magnitude filterbank.
This contrasts with the cumsum front end in \cref{tab:cumsum-frontend-hop400-app}, where real and imaginary channels help because the running cumsum magnitude is a noisier spectral estimate.

\begin{table}[ht]
\centering
\caption{Sin-cos convolutional filterbank front-end ablation at hop 160 samples. All rows use the same CNN backend.}
\label{tab:filterbank-frontend-app}
\small
\begin{tabular}{llrr}
\toprule
Model & Front-end features & Test & Params \\
\midrule
FilterbankSinCos Frozen & $\log(\sin^2+\cos^2)$, frozen mel filters & 95.7 & 25,628 \\
FilterbankSinCosMagReIm Frozen & $[\log |c|,\sin,\cos]\to$ Linear, frozen mel filters & 95.4 & 30,468 \\
FilterbankSinCos Learned & $\log(\sin^2+\cos^2)$, learned filters & 94.7 & 57,628 \\
FilterbankSinCosCombined Learned & $[\log |c|,\sin,\cos]$, learned filters & 94.5 & 61,468 \\
\bottomrule
\end{tabular}
\end{table}

\Cref{tab:mag-re-im-depth-app} reports the corresponding end-to-end cumsum sequence models that replace both the acoustic front end and the temporal backend with cumsum layers.
The same pattern appears there: log magnitude is a useful anchor, while carrying real and imaginary channels gives the next layer access to within-window phase.

\begin{table}[ht]
\centering
\caption{Mag+re+im ablation for raw-audio cumsum stacks. These models use cumsum front ends followed by cumsum temporal layers.}
\label{tab:mag-re-im-depth-app}
\small
\begin{tabular}{lrrl}
\toprule
Model & Test & Params & Feature path \\
\midrule
CumsumE2EMag & 93.5 & 62,732 & mag+log front end; complex cumsum layers \\
MagDeep V2 & 92.8 & 71,556 & $[\log |c|,\mathrm{Re}\ c,\mathrm{Im}\ c]$ all layers, with projection \\
MagDeep V2 & 92.5 & 63,812 & log-mag first layer, then $[\log |c|,\mathrm{Re}\ c,\mathrm{Im}\ c]$ \\
MagDeep V1 & 90.5 & 65,084 & log-mag only, wider model \\
MagDeep V1 & 89.2 & 33,692 & log-mag only, smaller model \\
\bottomrule
\end{tabular}
\end{table}

\bibliography{references}

@article{warden2018speechcommands,
  title={Speech Commands: A Dataset for Limited-Vocabulary Speech Recognition},
  author={Warden, Pete},
  journal={arXiv preprint arXiv:1804.03209},
  year={2018}
}

@inproceedings{zhang2017helloedge,
  title={Hello Edge: Keyword Spotting on Microcontrollers},
  author={Zhang, Yundong and Suda, Naveen and Lai, Liangzhen and Chandra, Vikas},
  booktitle={Proceedings of Machine Learning and Systems},
  year={2018}
}

@inproceedings{choi2019tcresnet,
  title={{Temporal Convolution for Real-time Keyword Spotting on Mobile Devices}},
  author={Choi, Seungwoo and Seo, Seokjun and Shin, Beomjun and Byun, Hyeongmin and Kersner, Martin and Kim, Beomsu and Kim, Dongyoung and Ha, Sungjoo},
  booktitle={Proc. Interspeech 2019},
  pages={3372--3376},
  year={2019},
  doi={10.21437/Interspeech.2019-1363}
}

@inproceedings{kim2021bcresnet,
  title={{Broadcasted Residual Learning for Efficient Keyword Spotting}},
  author={Kim, Byeonggeun and Chang, Simyung and Lee, Jinkyu and Sung, Dooyong},
  booktitle={Proc. Interspeech 2021},
  pages={4538--4542},
  year={2021},
  doi={10.21437/Interspeech.2021-383}
}

@inproceedings{rybakov2020streaming,
  title={{Streaming Keyword Spotting on Mobile Devices}},
  author={Rybakov, Oleg and Kononenko, Natasha and Subrahmanya, Niranjan and Visontai, Mirk{\'o} and Laurenzo, Stella},
  booktitle={Proc. Interspeech 2020},
  pages={2277--2281},
  year={2020},
  doi={10.21437/Interspeech.2020-1003}
}

@inproceedings{jose2022latency,
  title={{Latency Control for Keyword Spotting}},
  author={Jose, Christin and Wang, Joe and Strimel, Grant and Khursheed, Mohammad Omar and Mishchenko, Yuriy and Kulis, Brian},
  booktitle={Proc. Interspeech 2022},
  pages={1891--1895},
  year={2022},
  doi={10.21437/Interspeech.2022-10608}
}

@inproceedings{park2019specaugment,
  title={{SpecAugment}: A Simple Data Augmentation Method for Automatic Speech Recognition},
  author={Park, Daniel S and Chan, William and Zhang, Yu and Chiu, Chung-Cheng and Zoph, Barret and Cubuk, Ekin D and Le, Quoc V},
  booktitle={Proc. Interspeech 2019},
  pages={2613--2617},
  year={2019},
  doi={10.21437/Interspeech.2019-2680}
}

@inproceedings{vaswani2017attention,
  title={Attention is All You Need},
  author={Vaswani, Ashish and Shazeer, Noam and Parmar, Niki and Uszkoreit, Jakob and Jones, Llion and Gomez, Aidan N and Kaiser, Lukasz and Polosukhin, Illia},
  booktitle={Advances in Neural Information Processing Systems},
  volume={30},
  year={2017}
}

@article{su2021roformer,
  title={{RoFormer}: Enhanced Transformer with Rotary Position Embedding},
  author={Su, Jianlin and Lu, Yu and Pan, Shengfeng and Murtadha, Ahmed and Wen, Bo and Liu, Yunfeng},
  journal={arXiv preprint arXiv:2104.09864},
  year={2021}
}

@techreport{blelloch1990prefix,
  title={Prefix Sums and Their Applications},
  author={Blelloch, Guy E.},
  institution={School of Computer Science, Carnegie Mellon University},
  number={CMU-CS-90-190},
  year={1990}
}

@inproceedings{gu2021s4,
  title={Efficiently Modeling Long Sequences with Structured State Spaces},
  author={Gu, Albert and Goel, Karan and R{\'e}, Christopher},
  booktitle={Proc. International Conference on Learning Representations},
  year={2022}
}

@inproceedings{gupta2022dss,
  title={Diagonal State Spaces are as Effective as Structured State Spaces},
  author={Gupta, Ankit and Gu, Albert and Berant, Jonathan},
  booktitle={Advances in Neural Information Processing Systems},
  volume={35},
  year={2022}
}

@article{gu2022s4d,
  title={On the Parameterization and Initialization of Diagonal State Space Models},
  author={Gu, Albert and Gupta, Ankit and Goel, Karan and R{\'e}, Christopher},
  journal={arXiv preprint arXiv:2206.11893},
  year={2022}
}

@inproceedings{smith2023s5,
  title={Simplified State Space Layers for Sequence Modeling},
  author={Smith, Jimmy T. H. and Warrington, Andrew and Linderman, Scott W.},
  booktitle={Proc. International Conference on Learning Representations},
  year={2023}
}

@article{gu2023mamba,
  title={Mamba: Linear-Time Sequence Modeling with Selective State Spaces},
  author={Gu, Albert and Dao, Tri},
  journal={arXiv preprint arXiv:2312.00752},
  year={2023}
}

@inproceedings{zeghidour2021leaf,
  title={{LEAF}: A Learnable Frontend for Audio Classification},
  author={Zeghidour, Neil and Teboul, Olivier and de Chaumont Quitry, F{\'e}lix and Tagliasacchi, Marco},
  booktitle={Proc. International Conference on Learning Representations},
  year={2021}
}
\bibliographystyle{tmlr}

\end{document}